\documentclass[11pt]{article}
\usepackage{latexsym,amsthm,amsmath,amssymb}
\usepackage[utf8x]{inputenc}

\newtheorem{theorem}{Theorem}
\newtheorem{lemma}[theorem]{Lemma}
\newtheorem{definition}[theorem]{Definition}

\renewcommand{\title}[1]{{\noindent\bf\Large #1}}
\newcommand{\speaker}[1]{\noindent\underline{\bf #1} }
\newcommand{\coauthor}[1]{{\noindent\bf#1 }}
\newcommand{\affiliation}[1]{{\noindent\bf --- #1 }}

\newcommand{\II}[2]{[\![ #1 , #2 ]\!]}

\newcommand{\FunDef}[5]{
\begin{array}{ccccc}
#1 & : & #2 & \to & #3 \\
 & & #4 & \mapsto & #5 \\
\end{array}
}

\usepackage{graphicx}

\begin{document}
\title{Discrete Morse theory for the collapsibility of supremum sections}

\vspace{2ex}

%A coauthor, add as many as needed
\coauthor{Balthazar Bauer}
\affiliation{INRIA, DIENS, PSL research, CNRS, Paris, France}

% The speaker
\speaker{Lucas Isenmann}
\affiliation{LIRMM, Université de Montpellier, CNRS, Montpellier, France}

\vspace{2ex}

\begin{abstract}
    The Dushnik-Miller dimension of a poset $\le$ is the minimal number $d$ of linear extensions $\le_1, \ldots , \le_d$ of $\le$ such that $\le$ is the intersection of $\le_1, \ldots , \le_d$.
    Supremum sections are simplicial complexes introduced by Scarf~\cite{Scarf} and are linked to the Dushnik-Miller as follows: the inclusion poset of a simplicial complex is of Dushnik-Miller dimension at most $d$ if and only if it is included in a supremum section coming from a representation of dimension $d$.
    Collapsibility is a topological property of simplicial complexes which has been introduced by Whitehead~\cite{Whitehead} and which resembles shellability.
    While Ossona de Mendez~\cite{Ossona} proved that a particular type of supremum sections are shellable, we show in this article that supremum sections are in general collapsible thanks to the discrete Morse theory developped by Forman~\cite{Forman}.
\end{abstract}

\section{Introduction}

The order dimension (also known as the Dushnik-Miller dimension) of a poset $\le$ has been introduced by Dushnik and Miller \cite{DM}.
It is defined as the minimum number $d$ of linear extensions $\le_1, \ldots, \le_d$ of $\le$ such that $\le$ is the intersection of these extensions \textit{i.e.} $\forall x,y \in V, x \le y \iff (\forall i \in \II1d, x \le_i y)$.
See \cite{Trotter} for a comprehensive study of this topic.
This notion is important because, for example, of a theorem of Schnyder~\cite{Schnyder} which states that a graph is planar if and only if the Dushnik-Miller dimension of the inclusion poset of the associated simplicial complex is at most $3$.
Representations were introduced by Scarf~\cite{Scarf}.
A $d$-representation on a set $V$ is a set of $d$ linear orders on $V$.
Given $R$ a representation on a set $V$, we can define a simplicial complex $\Sigma(R)$ associated to this representation that we call its supremum section.
Scarf proved that every supremum section of a representation satisfying some additional properties, so called ``standard'', is the inclusion poset of a $d$-polytope with one face removed.
Ossona de Mendez~\cite{Ossona} proved that every abstract simplicial complex of Dushnik-Miller dimension at most $d$ is contained in a complex which is shellable and has a straight line embedding in $\mathbb{R}^{d-1}$.
Supremum sections also appeared in commutative algebra: Bayer \textit{et al.}~\cite{BPS} studied monomial ideals which are linked to supremum sections by what they call Scarf complexes.
They are used by Felsner \textit{et al.}~\cite{FK} in order to study orthogonal surfaces.
They also appear in the study of Gonçalves \textit{et al.}~\cite{GI} of a variant of Delaunay graphs and in the study of empty rectangles graphs by Felsner~\cite{Felsner}.
Furthermore, they also appear in spanning-tree-decompositions and in the box representations problem as shown by Evans \textit{et al.}~\cite{EFKU}.

% https://arxiv.org/abs/0907.2954

The goal of our article is to generalize the result of Ossona de Mendez about the shellability of standard supremum sections to every supremum sections.
As there exists supremum sections which are not shellable, for instance the simplicial complex characterized by its facets $\{a,b,c\}$ and $\{c,d,e\}$, we will replace shellability by collapsibility which is a similar notion.
A collapse is a topological operation on simplicial complexes, and more generally on CW-complexes, introduced by Whitehead~\cite{Whitehead} in order to define a simple homotopy equivalence which is a refinement of the homotopy equivalence.
A complex is said to be collapsible if it collapses to a point.
See \cite{Kozlov} for a comprehensive study of this topic.
The discrete Morse theory introduced by Forman~\cite{Forman} is based on this notion and has numerous applications in applied mathematics and computer science.
%voir https://www.sciencedirect.com/science/article/pii/S0012365X05002906#bib15 pour ça
Homotopy equivalence is a topological notion of topological spaces introduced to classify topological spaces.
Roughly speaking, two spaces are said to be homotopy equivalent if there exists a continuous deformation from one to the other.
A topological space is said to be contractible if it is homotopy equivalent to a point.
Collapsible spaces form an important subclass of contractible spaces.
While contractibility is algorithmically  undecidable by a result of Novikov~\cite{VKF}, the subclass of collapsible spaces is algorithmically  recognizable.
More precisely Tancer~\cite{Tancer} showed that it is NP-complete to decide whether a simplicial complex is collapsible.
Furthermore, every $1$-dimensional contractible complex is collapsible but the house with two rooms \cite{Bing} and the dunce hat \cite{Zeeman} show that there are complexes which are contractible but not collapsible.
Finally, the conjecture of Zeeman~\cite{Zeeman}, which implies the Poincarré conjecture, states that for every finite contractible $2$-dimensional CW-complex $K$, the space $K \times [0,1]$ is collapsible.

% There is a simple reason for appreciating collapsible objects: a collapsible (PL) n-manifold is always (PL) homeomorphic to a disc! (Although a contractible one may not, for instance in dimension 4.) For a proof, see [Rourke C.P., Sanderson B.J. Introduction to piecewise-linear topology (Springer, 1972)]. 

% https://link.springer.com/article/10.1007%2Fs10711-008-9231-7
% https://mathoverflow.net/questions/34703/is-the-topological-concept-of-collapsible-useful

% https://arxiv.org/pdf/1711.08436.pdf    Shellability NPC
% https://arxiv.org/pdf/1211.6254.pdf   Collapsibility NPC

\section{Notations}

In the following, $V$ is a finite set. 
An \textit{(abstract) simplicial complex} $\Delta$ is a subset of $\mathcal{P}(V)$ closed by inclusion (\textit{i.e.} $\forall X \in \Delta, \forall Y \subseteq X, Y \in \Delta$).
We call \textit{faces} the elements of $\Delta$ and \textit{facets} the maximal faces of $\Delta$ according to the inclusion order.

\begin{definition}[Ossona de Mendez~\cite{Ossona}]
    \label{definition:representation}
    Given a linear order $\le$ on a set $V$, an element $x \in V$, and a set $F \subseteq V$, we say that $x$ {\it dominates} $F$ in $\le$, and we denote it $F\le x$, if $f \le x$ for every $f \in F$. 
	A $d$-\textup{representation} $R$ on a set $V$ is a set of $d$ linear orders $\le_1 , \ldots , \le_d$ on $V$.	
	Given a $d$-representation $R$, an element $x \in V$, and a set $F \subseteq V$, we say that $x$ {\it dominates} $F$ in $R$ if $x$ dominates $F$ in some order $\le_i\in R$.
	We define $\Sigma(R)$ as the set of subsets $F$ of $V$ such that every $v \in V$ dominates $F$ in $R$.
	The set $\Sigma(R)$ is called the \textup{supremum section} of $R$.
\end{definition}

It is easy to show that if $R$ is a $d$-representation on a set $V$, then $\Sigma(R)$ is a simplicial complex.
An example is the following $3$-representation on $\{a,b,c,d,e\}$: $a <_1 b <_1 e <_1 d <_1 c$, $c <_2 b <_2 a <_2 d <_2 e$, and $e <_3 d <_3 c <_3 b <_3 a$.
The corresponding complex $\Sigma(R)$, depicted on the left of Figure~\ref{fig:the_fig}, is characterized by its facets $\{a,b\},
\{b,c,d\}$, and $\{b,d,e\}$.
For example $\{a,b,c\}$ is not in $\Sigma(R)$ as $b$ does not dominate $\{a,b,c\}$ in any order.

\begin{definition}
    Let $\Delta$ be a simplicial complex.
    We say that a face $F$ of $\Delta$ is a \textup{free face} of $\Delta$ if it is non-empty, non-maximal and contained in only one facet of $\Delta$.
    
    Let $\Delta$ and $\Gamma$ be two simplicial complexes.
    We say that $\Delta$ \textup{collapses} to $\Gamma$ if there exists $k$ simplicial complexes $\Delta_1, \ldots , \Delta_k$ and a free face $F_i$ of $\Delta_i$ for every $i \in \II1{k-1}$ such that $\Delta_1 = \Delta$ , $\Delta_{i+1} = \Delta_i \setminus \{ F \in \Delta_i :  F_i \subseteq F \}$ for every $i \in \II1{k-1}$ and $\Delta_k = \Gamma$. 
    We say that $\Delta$ is \textup{collapsible} if it collapses to a point.
\end{definition}

The Hasse diagram of a poset is the transitive reduction of the digraph of the poset.
Let $R$ be a representation on a set $V$, we denote $H(R)$ the Hasse diagram of the inclusion poset of $\Sigma(R)$.

\begin{definition}
    Let $(\le,V)$ be a poset and let $M$ be a matching of the Hasse diagram of $\le$.
    %A matching $M$ of the Hasse diagram of $\le$ is a subset of $V \times V$ such that $(a,b) \in M \Rightarrow a \le b$ and such that every element of $V$ appear at most $1$ time in a pair of the matching.
    For an arc $a$ of the Hasse diagram of $\le$, we denote $d(a)$ and $u(a)$ the elements of $V$ such that $a = (u(a),d(a))$ and $d(a) < u(a)$.
    A matching $M$ of the Hasse diagram of $\le$ is said to be \textup{acyclic} if, when reversing the orientation of the arcs of $M$,  the Hasse diagram remains acyclic.
\end{definition}

It is known that if $\le$ is the poset of inclusion of a simplicial complex and $M$ is a matching of the Hasse diagram of $\le$ then $M$ is acyclic if and only if there is no sequence of arcs $m_1, \ldots , m_n$ of $M$ such that $( u(m_{i+1}),d(m_i))$ is in the Hasse diagram for all $i \in \II1{n-1}$ as well as  $(u(m_1),d(m_n))$.

\begin{theorem}[Chari~\cite{Chari}]
    \label{theorem:CAM_colla}
    Let $\Delta$ be a simplicial complex.
    If the Hasse diagram of the inclusion poset of $\Delta$ admits a complete (\textup{i.e.} perfect) acyclic matching, then $\Delta$ is collapsible.
\end{theorem}

\section{Our contribution}

\begin{theorem}
    \label{theorem:main}
    Let $R$ be a representation on a set $V$.
    Then $\Sigma(R)$ is collapsible.
\end{theorem}

Because of Theorem~\ref{theorem:CAM_colla}, it is enough to show that if $R$ is a representation on a set $V$, then $H(R)$ admits a complete acyclic matching.

\subsection{Proofs}

The proof relies on an induction on the dimension of the representation $R$.
Let $R$ be a $d$-representation on a set $V$.
We denote $R' = (\le_1 , \ldots , \le_{d-1})$ the $(d-1)$-representation on $V$ obtained from $R$ by deleting the order $\le_d$.

\begin{lemma}
    The simplicial complex $\Sigma(R')$ is a subcomplex of $\Sigma(R)$.
\end{lemma}
\begin{proof}
	Let $F$ be a face of $\Sigma(R')$.
	As every element $x$ of $V$ dominates $F$ in at least one of the orders $\le_1, \ldots , \le_{d-1}$, the element $x$ also dominates $F$ in at least one of the orders $\le_1 , \ldots , \le_d$.
	We conclude that $F \in \Sigma(R)$.
\end{proof}

See Figure~\ref{fig:the_fig} to see what $\Sigma(R')$ is for the example representation of Definition~\ref{definition:representation}.
\begin{figure}[h]
    \centering
    \includegraphics{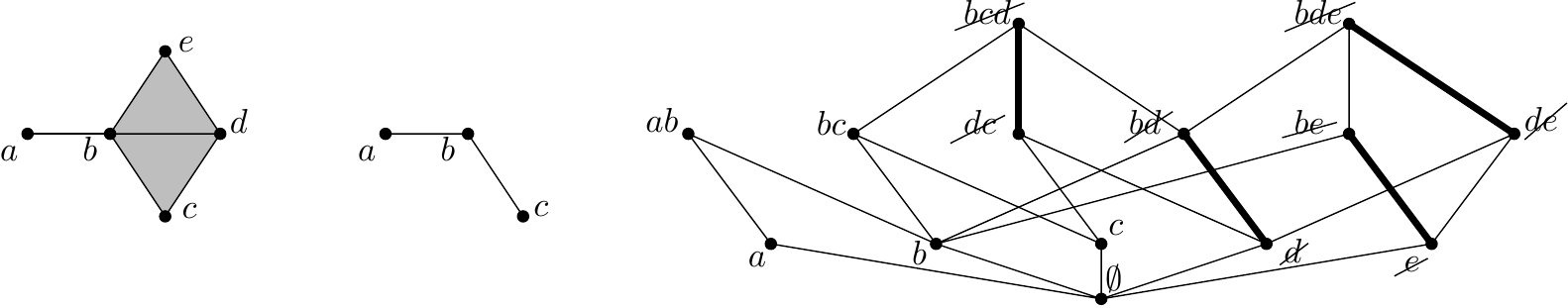}
    \caption{$R$ denotes the representation from the example. From left to right: $\Sigma(R)$ where a grey section corresponds to a face with $3$ elements, $\Sigma(R')$ where $R'$ is the representation obtained from $R$ by deleting the order $\le_3$, the Hasse diagram of $\Sigma(R)$ where the crossed-out faces are the faces from $\Sigma(R) \setminus \Sigma(R')$ and where the fat edges correspond to a complete acyclic matching of $\Sigma(R) \setminus \Sigma(R')$.}
    \label{fig:the_fig}
\end{figure}

\begin{lemma}
    We define the function $\psi$ by

    \[
	\FunDef{\psi}{\Sigma(R) \setminus \Sigma(R')}{V}{F}{\min_{<_d} \left\{x \in V : x <_i (\max_{\le_i} F) \; \; \forall i \in \II1{d-1} \right\} }
    \]
	Then the function $\psi$ is well-defined.
\end{lemma}
\begin{proof}
	Let $F$ be a face of $\Sigma(R) \setminus \Sigma(R')$.
	We denote $f_i = \max_{\le_i} F$ for every $i \in \II1d$.
	As $F \not\in \Sigma(R')$, there exists an element $x \in V$ such that $\forall i \in \II1{d-1}, x <_i f_i $.
	So the minimum is taken in a non-empty set.
\end{proof}

We define the sets  $A = \{ F \in \Sigma(R) \setminus \Sigma(R') : \psi(F) \not\in F \}$ and $B = \{ F \in \Sigma(R) \setminus \Sigma(R') : \psi(F) \in F \}$.
The goal is to find a complete acyclic matching between $A$ and $B$.

\begin{lemma}
    \label{lemma:psi_property}
    For every $F \in A$, we have $F \cup \{\psi(F)\} \in B$, $\psi(F \cup \{ \psi(F)\}) = \psi(F)$ and $\max_{\le_d} F < \psi(F)$.
    
    For every $F \in B$, we have $F \setminus \{\psi(F)\} \in A$ and $\psi(F \setminus \{ \psi(F)\}) = \psi(F)$.
\end{lemma}
\begin{proof}
	Let $F$ be in $A$, we denote $F' = F \bigcup \{\psi(F)\}$.
	For every $i \in \II1d$, we denote $f_i$ (resp. $f_i'$) the maximum of $F$ (resp. $F'$) in the order $\le_i$.
	By definition of $\psi$, $\psi(F) <_i f_i$ for every $i \in \II1{d-1}$.
	Furthermore, $f_d <_d \psi(F)$, otherwise $\psi(F)$ would not dominate $F$.
	Thus, $f_i = f_i'$ for every $i \in \II1{d-1}$ and $f_d <_d f_d' = \psi(F)$.
	Suppose that $F' \not\in \Sigma(R)$.
	Then there would exist $a$ such that $a$ does not dominate $F'$ in any order.
	Thus $a <_i f_i (=f_i')$ for every $i \in \II1{d-1}$ and $a <_d \psi(F)$ which contradicts the minimality of $\psi(F)$.
	We deduce that $F' \in \Sigma(R)$.
	
	As $\psi(F) <_i f_i'$ for every $i \in \II1{d-1}$, $\psi(F)$ does not dominate $F'$ in $R'$, $F' \not\in \Sigma(R')$ and $\psi(F') \le_d \psi(F)$.
	If $\psi(F') <_d \psi(F)$ then we would have $\psi(F') <_i f_i'$ for every $i \in \II1d$ as $f_d' = \psi(F)$.
	We deduce that $\psi(F') = \psi(F) = f_d' \in F'$.
	Finally, we conclude that $F \cup \{\psi(F)\} \in B$.
	
	The second property can be proved in the same manner.
\end{proof}

\begin{lemma}
    \label{lemma:CAM}
    The Hasse diagram of the inclusion poset of $\Sigma(R) \setminus \Sigma(R')$ admits a complete acyclic matching.
\end{lemma}
See Figure~\ref{fig:the_fig} to see an example of a complete acyclic matching.
\begin{proof}
    We define the function $\varphi : A \to B$ defined by $\varphi(F) = F \cup \{ \psi(F)\}$ for every $F \in A$.
    %\[
    %\FunDef{\varphi}{A}{B}{F}{F \cup \{ \psi(F) \}}
    %\]
	Let us show that $\varphi$ is a bijection.
	To do so, we define the function $\eta : B \to A$ by $\eta(F) = F \setminus \{\psi(F)\}$ where $F \in B$.
	Lemma~\ref{lemma:psi_property} implies that $\eta$ is well defined, that $\eta \circ \varphi = \mathsf{id}_A$, and that $\varphi \circ \eta = \mathsf{id}_B$.
	Thus $\varphi$ is a bijection and $\varphi$ defines a complete matching $M = \{ (F,\varphi(F)) : F \in A \}$ between $A$ and $B$.
	
	%Let us show that $M$ is acyclic.
	%Towards this goal, we show that for every $F \in A$, for every $f \in F$, $f <_d \psi(F)$.
	%Let $F \in A$, then $F \in \Sigma(R) \setminus \Sigma(R')$ and $\psi(F) \not\in F$.
	%We denote $f_i = \max_{\le_i} F$ for every $i \in \II1d$.
	%By definition of $\psi$, $\psi(F) <_i f_i$ for every $i \in \II1{d-1}$.
	%As $F \in \Sigma(R)$, $\psi(F)$ dominates $F$ in $R$ which is only possible in the order $\le_d$.
	%Therefore $f_d \le_d \psi(F)$.
	%As $\psi(F) \not\in F$, we have $f <_d \psi(F)$ for every $f \in F$.
	
	Suppose that $M$ is not acyclic: there exists a sequence $m_1, \ldots , m_n$ of arcs of $M$ where  $m_i = (F_i \cup \{\psi(F_i)\}, F_i)$ for a $F_i \in A$ for every $i \in \II1n$ such that $(F_{i+1} \cup \{\psi(F_{i+1})\}),F_i)$ is in the Hasse diagram for every $i \in \II1{n-1}$ as well as $(F_1 \cup \{\psi(F_1)\},F_n)$.
    As for every $i \in \II1{n-1}$, $F_i \subseteq F_{i+1} \cup \psi(F_{i+1})$ and $|F_i| +1 = |F_{i+1} \cup \psi(F_{i+1})|$, we deduce that $\psi(F_{i+1}) \in F_i$.
    Therefore $\psi(F_{i+1}) <_d \psi(F_i)$ for every $i \in \II1{n-1}$ and thus $\psi(F_n) <_d \psi(F_1)$.
    As $(F_1 \cup \{\psi(F_1)\},F_n )$ is in the Hasse diagram, we show in the same way that $\psi(F_1) <_d \psi(F_n)$ which contradicts the fact that $\psi(F_n) <_d \psi(F_1)$.
    We conclude that $M$ is a perfect acyclic matching of $\Sigma(R)\setminus \Sigma(R')$.
\end{proof}

We can now prove Theorem~\ref{theorem:main}.
\begin{proof}
	We prove the result by induction on the number of orders.
	Let $R = (\le_1)$ be a $1$-representation on $V$.
	We denote $m_1$ the minimum on $V$ in $\le_1$.
	Let $F$ be a face of $\Sigma(R)$ which contains an element $x$ different from $m_1$.
	Then $m_1$ does not dominate $F$ in $R$ as $m_1 <_1 x$.
	The set $\{m_1\}$ is a face of $\Sigma(R)$ as every element of $V$ dominates $\{m_1\}$ in the order $\le_1$.
	Thus $\Sigma(R) = \{\emptyset, \{m_1\} \}$ and $H(R)$ admits a complete acyclic matching $(\{m_1\},\emptyset)$.
	The base case is therefore true.
	
	Let $d \ge 2$, we now suppose that the result is true for any $(d-1)$-representation on $V$.
	Let $R= (\le_1 , \ldots , \le_d)$ be a $d$-representation on $V$.
	We denote $R' = (\le_1 , \ldots , \le_{d-1})$ the $(d-1)$-representation on $V$ obtained from $R$ by deleting the order $\le_d$.
    We define $K$ as $\Sigma(R) \setminus \Sigma(R')$.
    Because of Lemma~\ref{lemma:CAM}, the Hasse diagram of the inclusion poset of $K$ admits a complete acyclic matching $M_1$.

    By induction hypothesis, $H(R')$ admits a complete acyclic matching $M_2$.
    Thus $M_1 \cup M_2$ is a complete matching of $H(R)$.
    Furthermore, $H(R)$ is the union of $H(R')$ and the Hasse diagram of the inclusion poset of $K$ with some arcs between $K$ and $\Sigma(R')$.
    If an arc between $K$ and $\Sigma(R')$ is oriented from $\Sigma(R')$ to $K$, then there would be a face $F$ of $\Sigma(R')$ that would contain a face $G$ of $K$.
    As $\Sigma(R')$ is closed by inclusion, then $G$ is also in $\Sigma(R')$ which contradicts the definition of $K$.
    Therefore the arcs between $K$ and $\Sigma(R')$ are oriented from $K$ to $\Sigma(R')$ and we deduce that $M_1 \cup M_2$ is a perfect acyclic matching of $H(R)$.
    We conclude by induction.
\end{proof}

\end{document}